\documentclass[12pt]{amsart}
\usepackage{amscd,color}
\usepackage{bbm}
\usepackage{mathrsfs}
\usepackage{bbm}
\usepackage{amsfonts}
\usepackage{amsmath, amssymb, amsthm}
\usepackage{stmaryrd}
\usepackage{bm}
\usepackage{graphics}
\usepackage{lmodern}                           
\allowdisplaybreaks

\usepackage{latexsym,amsfonts,amssymb}
\theoremstyle{plain}
\newtheorem{theorem}{Theorem}
\newtheorem{Main theorem}[theorem]{Main Theorem}
\newtheorem{lemma}{Lemma}[section]
\newtheorem{corollary}[lemma]{Corollary}
\newtheorem{proposition}[lemma]{Proposition}
\theoremstyle{definition}
\newtheorem{definition}[lemma]{Definition}
\newtheorem{remark}[lemma]{Remark}

\usepackage[colorlinks=true, allcolors=blue]{hyperref}

\newcommand{\ang}[1]{\left\langle #1 \right\rangle}

\begin{document}

\title{BiHom-Lie brackets and the Toda equation}

\author{Botong Gai$^1$}\address{$^1$School of Mathematics, Southeast University, Nanjing 211189, China}\email{230228425@seu.edu.cn}

\author{Chuanzhong Li$^{2\ 3}$}\address{$^2$College of Mathematics and Systems Science, Shandong University of Science and Technology,
Qingdao 266590, China} \address{$^3$Material Science, Innovation and Modelling (MaSIM) Research Focus Area, North-West University, Mafikeng Campus, Mmabatho 2735, South Africa}\email{lichuanzhong@sdust.edu.cn}

\author{Jiacheng Sun$^1$}\address{$^1$School of Mathematics, Southeast University, Nanjing 211189, China}\email{220242018@seu.edu.cn}

\author{Shuanhong Wang$^4$}\address{$^4$Shing-Tung Yau Center, School of Mathematics, Southeast University, Nanjing 210096, China}\email{shuanhwang@seu.edu.cn}
 
\author{Haoran Zhu$^5$}\address{$^5$School of Physical and Mathematical Sciences, Nanyang Technological University, 21 Nanyang Link, 637371, Singapore}\email{haoran.zhu@ntu.edu.sg}

\date{\today}

\begin{abstract}
We introduce a BiHom-type skew-symmetric bracket on  $\mathfrak{gl}(V)$ built from two commuting
inner automorphisms $\alpha=Ad_\psi$ and $\beta=Ad_\phi$ with $\psi,\phi\in \mathfrak{gl}(V)$
and integers $i,j$. We prove that
$(\mathfrak{gl}(V),[\cdot,\cdot]^{(i,j)}_{(\psi,\phi)},\alpha,\beta)$ is a BiHom--Lie algebra,
and we study the Lax equation obtained by replacing the commutator in the finite nonperiodic Toda lattice by this bracket.
For the symmetric choice $\phi=\psi$ with $(i,j)=(0,0)$, the deformed flow is equivariant under conjugation and becomes
gauge-equivalent, via $\widetilde L=\psi^{-1}L\psi$, to a Toda-type Lax equation with a conjugated triangular projection.
In particular, scalar deformations amount to a constant rescaling of time.
On embedded $2\times2$ blocks, we derive explicit trigonometric and hyperbolic formulas that make symmetry constraints
\textup{(e.g.\ tracelessness)} transparent. In the asymmetric hyperbolic case, we exhibit a trace obstruction showing
that the right-hand side is generically not a commutator, which amounts to symmetry breaking of the isospectral property.
We further extend the construction to the weakly coupled Toda lattice with an indefinite metric and provide explicit
$2\times2$ solutions via an inverse-scattering calculation, clarifying and correcting certain formulas in the literature.
The classical Toda dynamics are recovered at special parameter values.
\end{abstract}

\maketitle

\section{Introduction}

As one of the earliest examples of a nonlinear, completely integrable system, the Toda lattice hierarchy has attracted sustained interest from both mathematicians and physicists. After the lattice had been shown to admit solutions
expressible via elliptic functions, it was soon observed to possess multi-soliton solutions with
elastic collisions. This parallel with the KdV equation, together with the subsequent Lax-pair formulation, led to its recognition as an integrable system. Many mathematical treatments and generalizations of the original structures continue. In particular, one common thread is the relation of infinite-dimensional Hamiltonian and Lagrangian dynamical systems to infinite-dimensional Lie algebras~\cite{B, DKV, DVY}. This algebraic viewpoint, where invariance, equivariance, and symmetry constraints organize the dynamics, is also the focus of the present paper.

A Lie bracket is a basic instrument for encoding the algebraic content behind nonlinear dynamics. Therefore, deforming the bracket is a natural way to organize non-standard evolutions and to track symmetry versus symmetry breaking. In the context of conformal field theory (CFT) and quantum algebra~\cite{MN,SWZ,Y1,Y2,Y3}, Hom-~\cite{SWZZ} and BiHom-type~\cite{GMM} structures arise naturally, for instance in deformations of Virasoro-type algebras. Motivated by this, we employ a BiHom-type deformation that twists the commutator by two commuting inner automorphisms; see \cite{HLS,GMM,JMS,SX} for background and \cite{HYZ1,HYZ2,HYZ3, MN,SWZ,Y1,Y2,Y3} for related appearances in quantum algebra and high-energy physics.

In light of this, we present a deformed (non-standard) version of the finite nonperiodic Toda lattice obtained by replacing the commutator in the Lax equation with a BiHom-type bracket produced by the Makhlouf--Silvestrov construction. Concretely, we work with two commuting one-parameter families
\[
\psi,\phi: R_1\longrightarrow \mathfrak{gl}(V),
\]
and set $\alpha_s:=Ad_{\psi(s)}$, $\beta_s:=Ad_{\phi(s)}$. For two integers $i$ and $j$, we determine the unique skew-symmetric bilinear form within a natural ansatz that makes the BiHom--Jacobi identity hold. This bracket both generalizes and, for $i+j\neq 1$,  extends beyond the Hom-type bilinear form considered in~\cite{SX}. Beyond the abstract construction, we analyze three explicit deformation families (i.e.\ scalar dilations, planar rotations, and planar hyperbolic rotations) and compute the resulting \(2\times2\) block dynamics in detail, with particular attention to tracelessness, symmetry/equivariance and (non-) isospectrality indicators. 

To further extend the scope and, following the method of Li and He~\cite{L2}, we formulate the weakly coupled Toda lattice with an indefinite metric and its non-standard counterpart. In the two-dimensional case, we combine a direct Lax computation with the inverse-scattering scheme in~\cite{L3} to obtain unified explicit solutions and to clarify and correct certain formulas in \cite{BL}.

The layout of the paper is as follows. In Section~\ref{sect2}, we fix notation, recall the BiHom--Jacobi identity and define the deformed bracket. We also establish the uniqueness of the skew-symmetric operator within the bilinear ansatz. In Section~\ref{sect3}, we describe, in elementary group-theoretic terms, three concrete one-parameter deformations that will be used later. In Section~\ref{sect4}, we specialize the deformed Lax equation to three families (scalar, rotation, hyperbolic rotation), extract their structural consequences, and provide closed \(2\times2\) formulas that make the underlying symmetry constraints explicit. Section~\ref{sect5} gives a Miura-type relation clarifying when the deformed flow is conjugate to a Toda-type Lax equation and when asymmetry leads to symmetry breaking. Section~\ref{sect6} treats the weakly coupled case with indefinite metrics, deriving explicit solutions in the two-dimensional normal form and explaining the corrections mentioned above.

\section{Deformations of the Lie bracket}\label{sect2}

In this section we define the deformed (BiHom) Lie bracket which will be used throughout, and then we also give the corresponding Lax equation for the Toda lattice. 

Let $V=\mathbb{R}^n$ and let $\mathfrak{gl}(V)$ be the space of all real $n\times n$ matrices.
For $L\in \mathfrak{gl}(V)$ we denote by $L_{>0}$ (resp.\ $L_{<0}$) the strictly upper (resp.\ strictly lower) triangular part of $L$, and set
\[
B(L):=L_{>0}-L_{<0}.
\]
Note that $B(L)$ is skew-symmetric with respect to the upper/lower triangular splitting in the sense that $B(L)^\top=-B(L)$ whenever $L$ is symmetric. The skew-symmetry will be mirrored by the BiHom-type deformation below.

We consider two $C^1$ (continuously differentiable) one-parameter maps
\[
\psi,\phi:R_1\longrightarrow \mathfrak{gl}(V),\qquad s\mapsto \psi(s),\ \phi(s),
\]
where $R_1\subseteq \mathbb{R}$ is an interval, such that the matrices $\psi(s)$ and $\phi(s)$ commute for each $s\in R_1$.
We write $\alpha_s:=\operatorname{Ad}_{\psi(s)}$ and $\beta_s:=\operatorname{Ad}_{\phi(s)}$ for their adjoint actions on $\mathfrak{gl}(V)$.
For integer exponents, we use the
convention $\psi(s)^{-1}=\psi(s)^{\,(-1)}$, where the former denotes the inverse of $\psi(s)$ and the latter
denotes the exponent -1. The same convention applies to~$\phi(s)$.

For completeness, we recall the Hom-Jacobi identity (\cite{HLS,JMS}) in terms of a single twisting map $\alpha$:
\begin{equation}\label{twistjacobi}
[\alpha(a),[b,c]]+[\alpha(b),[c,a]]+[\alpha(c),[a,b]]=0.
\end{equation}
Further, a \emph{BiHom-Lie algebra} (\cite{GMM}) over a field $k$ is a quadruple $(L, [\ , \ ], \alpha, \beta)$, where $L$ is a $k$-linear space, $\alpha: L\rightarrow L$, $\beta: L\rightarrow L$, $[\ , \ ]: L\otimes L\rightarrow L$ are linear maps,  with the notation $[\ , \ ](a\otimes b)=[a, b]$,  satisfying the following conditions, for all $a, b, c\in L$:
\begin{equation}
\alpha\circ\beta=\beta\circ\alpha,\label{M1}
\end{equation}
\begin{equation}
[\beta(a), \alpha(b)]=-[\beta(b), \alpha(a)],\label{M2}
\end{equation}
\begin{equation}
\alpha([a, b])=[\alpha(a), \alpha(b)],\qquad \beta([a, b])=[\beta(a), \beta(b)],\label{M3}
\end{equation}
\begin{equation}\label{twistjacobi2}
[\beta^2(a),[\beta(b),\alpha(c)]]+[\beta^2(b),[\beta(c),\alpha(a)]]+[\beta^2(c),[\beta(a),\alpha(b)]]=0,
\end{equation}
In what follows we construct an explicit bracket $[\cdot,\cdot]^{(i,j)}_{(\psi,\phi)}$ which satisfies \eqref{twistjacobi2} with respect to $(\alpha_s,\beta_s)$ and exhibits the required BiHom \emph{skew-symmetry} between the $\psi$-actions and $\phi$-actions (see Appendix~\ref{fulu1} for a uniqueness statement within a bilinear ansatz).

Fix $s\in R_1$ and integers $i,j\in\mathbb Z$.
For $A,B\in\mathfrak{gl}(V)$, set
\begin{equation}\label{def_bihom_bracket}
[A,B]^{(i,j)}_{(\psi,\phi)}(s)
:= \psi(s)\,A\,\psi(s)^{\,i}\phi(s)^{\,j}\,B\,\phi(s)^{-1}
   \;-\; \phi(s)\,B\,\psi(s)^{\,i+2}\phi(s)^{\,j-2}\,A\,\psi(s)^{-1},
\end{equation}
which is the form used in the uniqueness derivation in Appendix~\ref{fulu1}. We always suppress the parameter $s$ from the notation without confusion. Based on this, we state the following theorem.

\begin{theorem}\label{thm:bihom_lie}
Let $\psi(s),\phi(s)\in \mathfrak{gl}(V)$ commute for each $s\in R_1$ and  $i,j\in\mathbb{Z}$.
Then the quadruple
\[
\bigl(\mathfrak{gl}(V),[\cdot,\cdot]^{(i,j)}_{(\psi,\phi)}(s),\alpha_s,\beta_s\bigr)
\]
is a BiHom-Lie algebra.
\end{theorem}

Within the bilinear ansatz considered in Appendix~\ref{fulu1}, the formula \eqref{def_bihom_bracket} is characterized by the choice of $i$ and $j$.

\begin{remark}[Specializations]\label{rem:specializations}
\leavevmode
\begin{itemize}
  \item If $\psi(s)=\phi(s)$, then
  \[
  [A,B]^{(i,j)}_{(\psi,\psi)}(s)
  =\psi(s) A\,\psi(s)^{\,i+j}B\,\psi(s)^{-1}-\psi(s) B\,\psi(s)^{\,i+j}A\,\psi(s)^{-1}.
  \]
  In particular, for $(i,j)=(-1,0)$, this simplifies to
  \[
  [A,B]^{(-1,0)}_{(\psi,\psi)}(s)
  =\psi(s) A\psi(s)^{-1}\,B\,\psi(s)^{-1}-\psi(s) B\psi(s)^{-1}\,A\,\psi(s)^{-1},
  \]
  which is \emph{not} equal to $\psi(s)[A,B]\psi(s)^{-1}$ in general. This case has been discussed in \cite{SX}.
  
  Specially, for $(i,j)=(0,0)$, this reduces to
  \[
  [A,B]^{(0,0)}_{(\psi,\psi)}(s)
  =\psi(s) AB\,\psi(s)^{-1}-\psi(s) BA\,\psi(s)^{-1},
  \]
  which is equal to $\psi(s)[A,B]\psi(s)^{-1}$. At the level of the Lax equation, however, the resulting flow is gauge-equivalent to a Toda-type equation with a conjugated triangular projection; see Section~\ref{sect5}.
  \item If $\psi(s)=\phi(s)=\mathrm{Id}_{V}$, we recover the ordinary commutator.
\end{itemize}
\end{remark}

Following from \cite{F,KM,KY}, the (finite nonperiodic) Toda lattice is written as the Lax equation
\begin{equation}\label{originalSys}
\frac{\mathrm{d}}{\mathrm{d}t}L=[B(L),L],\qquad B(L):=L_{>0}-L_{<0}.
\end{equation}
Here $B(L)$ realizes the classical antisymmetrization with respect to the upper/lower triangular decomposition.

Fix $s\in R_1$.
Replacing the commutator in \eqref{originalSys} by the BiHom-type bracket \eqref{def_bihom_bracket}, we obtain the \emph{deformed (BiHom–Toda) Lax equation}
\begin{equation}\label{deformedSys}
\frac{\mathrm{d}}{\mathrm{d}t}L
= [B(L),L]^{(i,j)}_{(\psi,\phi)}(s),
\end{equation}
that is
\begin{equation}\label{def:BiHomToda}
\frac{\mathrm{d}}{\mathrm{d}t}L
= \psi(s) B(L)\,\psi(s)^{\,i}\phi(s)^{\,j}L\,\phi(s)^{-1}
  - \phi(s) L\,\psi(s)^{\,i+2}\phi(s)^{\,j-2}B(L)\,\psi(s)^{-1}.
\end{equation}
Here $s$ is a deformation parameter. Two important special cases are worth recording:
\begin{itemize}
  \item If $\psi(s)=\phi(s)$ and $(i,j)=(0,0)$, then the flow is gauge-equivalent to a Toda-type equation with a conjugated triangular projection (see Section~\ref{sect5} for the precise statement and proof).
  \item If $\psi(s)=\phi(s)=\mathrm{Id}_{V}$, \eqref{deformedSys} reduces to \eqref{originalSys}.
\end{itemize}
We will build on \eqref{deformedSys} in the subsequent sections.

\section{One-parameter deformations via conjugations}\label{sect3}

This section records the concrete one-parameter families of matrices whose adjoint actions we will use to build the deformed bracket from Section~\ref{sect2} and the corresponding deformed Toda flows.

For a $2\times 2$ matrix $M$, we write
\begin{equation}\label{eq:embedding_E}
\mathcal{E}(M):=\mathrm{diag}(I_{n-2},M)\in \mathfrak{gl}(\mathbb{R}^n),
\end{equation}
so that conjugation by $\mathcal{E}(M)$ acts nontrivially only on the last $2\times 2$ block and trivially on the $(n-2)\times(n-2)$ identity block. 

We shall use the following commuting one-parameter families $\psi,\phi:R_1\to \mathfrak{gl}(\mathbb{R}^n)$.

\begin{itemize}
  \item \textbf{(Scalar dilations).} \label{item:scalar_dilations}
  For $r\in \mathbb{R}^\times$ and $p,q\in\mathbb{R}$, set
  \begin{equation}\label{eq:scalar_psi_phi}
  \psi(r):=r^{p} I_n,\qquad \phi(r):=r^{q} I_n .
  \end{equation}
  Then $\psi(r_1)\psi(r_2)=\psi(r_1r_2)$, $\psi(r)^{-1}=\psi(r^{-1})$, and similarly for $\phi$; moreover $\mathrm{Ad}_{\psi(r)}=\mathrm{Id}_{\mathbb{R}^n}$ on $\mathfrak{gl}(\mathbb{R}^n)$. This family reflects a \emph{uniform scaling symmetry} of the bracketed dynamics.

  \item \textbf{(Planar rotations (Euclidean case)).} \label{item:euclidean_rotations}
  For $\theta\in\mathbb{R}$, define
  \begin{equation}\label{eq:R2_def}
  R_2(\theta):=\begin{pmatrix}\cos\theta&-\sin\theta\\ \sin\theta&\cos\theta\end{pmatrix},\qquad
  \psi(\theta):=\mathcal{E}\!\bigl(R_2(\theta)\bigr),\qquad
  \phi(\theta):=\mathcal{E}\!\bigl(R_2(\pm\theta)\bigr).
  \end{equation}
  We have $R_2(\theta_1)R_2(\theta_2)=R_2(\theta_1+\theta_2)$ and $R_2(\theta)^{-1}=R_2(-\theta)$, hence $\psi(\theta_1)\psi(\theta_2)=\psi(\theta_1+\theta_2)$; $\psi(\theta)$ commutes with $\phi(\theta)$ for either sign choice. We distinguish the \emph{symmetric} choice $\phi(\theta)=\psi(\theta)$ and the \emph{asymmetric} choice $\phi(\theta)=\mathcal{E}(R_2(-\theta))$:
  \begin{itemize}
    \item \emph{Symmetric choice} $\phi(\theta)=\psi(\theta)$.\label{choice:rot_symmetric} For $i,j\in\mathbb Z$, one has the identity
    \begin{equation}\label{eq:rot_symmetry_conjugacy}
    [A,B]^{(i,j)}_{(\psi,\psi)}(\theta)
    =\psi(\theta)\,A\,\psi(\theta)^{\,i+j}\,B\,\psi(\theta)^{-1}
     -\psi(\theta)\,B\,\psi(\theta)^{\,i+j}\,A\,\psi(\theta)^{-1}.
    \end{equation}
    In particular, for $(i,j)=(0,0)$, this reduces to
    \[
    [A,B]^{(0,0)}_{(\psi,\psi)}(\theta)
    =\psi(\theta)AB\,\psi(\theta)^{-1}
     -\psi(\theta)BA\,\psi(\theta)^{-1}.
    \]
    The associated deformed Lax flow is \emph{gauge-equivalent} to a Toda-type equation with a conjugated triangular projection, see Theorem~\ref{thm:gauge_equiv}. It coincides with the classical Toda flow precisely when $\psi(\theta)$ preserves the upper/lower triangular splitting (e.g.\ when $\psi(\theta)$ is diagonal).
    \item \emph{Asymmetric choice} $\phi(\theta)=\mathcal{E}(R_2(-\theta))$.\label{choice:rot_asymmetric} This produces a genuinely different bracket whose $2\times 2$-block effect will be analyzed in Section~\ref{sect4}.
  \end{itemize}

  \item \textbf{(Planar hyperbolic rotations (Lorentzian case)).} \label{item:hyperbolic_rotations}
  For $\lambda\in\mathbb{R}$, define
  \begin{equation}\label{eq:H2_def}
  H_2(\lambda):=\begin{pmatrix}\cosh\lambda&\sinh\lambda\\ \sinh\lambda&\cosh\lambda\end{pmatrix},\qquad
  \psi(\lambda):=\mathcal{E}\!\bigl(H_2(\lambda)\bigr),\qquad
  \phi(\lambda):=\mathcal{E}\!\bigl(H_2(\pm\lambda)\bigr).
  \end{equation}
  Then $H_2(\lambda_1)H_2(\lambda_2)=H_2(\lambda_1+\lambda_2)$ and $H_2(\lambda)^{-1}=H_2(-\lambda)$, so the additivity and invertibility properties hold for $\psi,\phi$ as in~\eqref{eq:R2_def}. Again $\psi(\lambda)$ commutes with $\phi(\lambda)$ for either sign choice. We likewise distinguish:
  \begin{itemize}
    \item \emph{Symmetric choice} $\phi(\lambda)=\psi(\lambda)$.\label{choice:hyp_symmetric}
    In particular, for $(i,j)=(0,0)$, the deformed Lax flow is gauge-equivalent to a Toda-type equation with a conjugated triangular projection (see Theorem~\ref{thm:gauge_equiv}); it coincides with the classical Toda flow precisely when $\psi(\lambda)$ preserves the upper/lower triangular splitting (e.g.\ $\psi(\lambda)$ is diagonal).
    \item \emph{Asymmetric choice} $\phi(\lambda)=\mathcal{E}(H_2(-\lambda))$.\label{choice:hyp_asymmetric} This leads to non-conjugate deformations; see Section~\ref{sect4} for $2\times2$ block formulas and trace properties.
  \end{itemize}
\end{itemize}

Recalling the bracket from \eqref{def_bihom_bracket}, the previous families induce the following specializations that will be used repeatedly:

\begin{itemize}
  \item[(A$^\prime$)] \emph{Scalar dilations.}\label{Aprime}
  If $\psi(r)=r^{p}I_n$ and $\phi(r)=r^{q}I_n$, then
  \begin{equation}\label{eq:scalar_factor}
  [A,B]^{(i,j)}_{(\psi,\phi)}(r)=r^{\,p(i+1)+q(j-1)}\,[A,B].
  \end{equation}
  Consequently, in the deformed Lax equation $\dot L=[B(L),L]^{(i,j)}_{(\psi,\phi)}(r)$, the right-hand side is a constant multiple of the classical one.

  \item[(B$^\prime$)] \emph{Rotations.}\label{Bprime}
  With $\psi(\theta)=\mathcal{E}(R_2(\theta))$ and either $\phi(\theta)=\psi(\theta)$ or $\phi(\theta)=\mathcal{E}(R_2(-\theta))$, formula \eqref{def_bihom_bracket} specializes to explicit trigonometric combinations. In the symmetric case with $(i,j)=(0,0)$, the associated Lax flow is gauge-equivalent to a Toda-type equation with conjugated projection (Theorem~\ref{thm:gauge_equiv}); see Section~\ref{sect4} for $2\times2$ block formulas in the asymmetric case.

  \item[(C$^\prime$)] \emph{Hyperbolic rotations.}\label{Cprime}
  The situation is analogous to Section \ref{Bprime}. The identities for $\cosh$ and $\sinh$ guarantee that the required manipulations carry over verbatim. In the symmetric case with $(i,j)=(0,0)$, the flow is gauge-equivalent to a Toda-type equation with conjugated projection; the asymmetric choice provides non-conjugate deformations discussed in Section~\ref{sect4}.
\end{itemize}

\begin{remark}\label{rmk:anchor_identity}
In all three families we have $\psi(0)=\phi(0)=I_n$ (or $\psi(1)=\phi(1)=I_n$ in the multiplicative scalar case \eqref{eq:scalar_psi_phi}). Thus the parameter $s$ controls a deformation anchored at the identity, preserving the natural \emph{symmetry at zero deformation} required by the setup in Section~\ref{sect2}.
\end{remark}

\section{Behavior and solutions in three canonical deformations}\label{sect4}

We now specialize the deformed Lax equation \eqref{deformedSys} to the three one-parameter families recorded in Section~\ref{sect3} and extract consequences for the dynamics. Throughout, $i,j\in\mathbb{Z}$ and $B(L)=L_{>0}-L_{<0}$. Lengthy $2\times 2$ block computations are deferred to Appendix~\ref{app:2x2}; here we state the structural facts that will be used later.


\subsection{Case I: Scalar dilations}\label{subsec:scalar}

In this case, the deformation only induces constant time rescaling of the classical Toda flow. All conserved quantities are preserved, and the phase portrait remains identical to the classical system---no changes to symmetry, integrability, or qualitative dynamics beyond adjusting the speed of evolution.

Let $\psi(r)=r^{p}I_n$, $\phi(r)=r^{q}I_n$ with $r\in\mathbb R^\times$.
By \eqref{def_bihom_bracket} one has
\begin{equation}\label{eq:scalar_multiple}
[B(L),L]^{(i,j)}_{(\psi,\phi)}(t)=\alpha\,[B(L),L](t),\qquad
\alpha:=r^{\,p(i+1)+q(j-1)}.
\end{equation}

\begin{proposition}\label{prop:time_rescaling}
Let $L_{\mathrm{Toda}}(t)$ be the solution to the classical Toda flow $\dot L=[B(L),L](t)$ with initial datum $L(0)=L_0$. Then the solution to the deformed flow
\[
\dot{\widehat{L}}=\frac{dL(\alpha^{-1}t)}{dt}=[B(L),L]^{(i,j)}_{(\psi,\phi)}(\alpha^{-1}t)\alpha^{-1}
\]
with the same initial datum is
\[
\widehat{L}(t)=L_{\mathrm{Toda}}(\alpha^{-1}\,t),\qquad \alpha=r^{\,p(i+1)+q(j-1)}.
\]
In particular, all classical conserved quantities (e.g.\ $\mathrm{tr}(L^k)$) are preserved, and the phase portrait is unchanged up to a constant time rescaling.
\end{proposition}

\begin{proof}
Equation \eqref{eq:scalar_multiple} gives $\dot L=\alpha\,[B(L),L](t)$. Then we have
$$
\frac{d\widehat{L}}{dt}=\frac{dL(\alpha^{-1} t)}{dt}=\alpha[B( L), L](\alpha^{-1}t)\alpha^{-1}=[B( L), L](\alpha^{-1}t),
$$ 
with $\widehat{L}(0)=L_0$. Hence $\widehat{L}(t)=L_{\mathrm{Toda}}(\alpha^{-1}t)$ is the solution.
\end{proof}

\subsection{Case II: Planar rotations}\label{subsec:rotations}

In this case, let $\psi(\theta)=\mathcal{E}(R_2(\theta))$, two subcases---symmetry and asymmetry are analyzed. We claim the flow is gauge-equivalent to classical Toda and preserves tracelessness and isospectrality for symmetric choice with $(i,j)=(0,0)$.

\paragraph{\textnormal{(a) Symmetric choice}}\label{par:rot-sym} $\phi(\theta)=\psi(\theta)$.
Then \eqref{def_bihom_bracket} gives, for $i,j\in\mathbb N$,
\begin{equation}\label{eq:rot_sym_bracket}
[B(L),L]^{(i,j)}_{(\psi,\psi)}(\theta)
=
\psi(\theta)\,B(L)\,\psi(\theta)^{i+j}\,L\,\psi(\theta)^{-1}
-
\psi(\theta)\,L\,\psi(\theta)^{i+j}\,B(L)\,\psi(\theta)^{-1}.
\end{equation}
In particular, when $(i,j)=(0,0)$, one can factor the outer conjugations as
\begin{equation}\label{eq:rot_sym_fact}
[B(L),L]^{(0,0)}_{(\psi,\psi)}(\theta)
=\bigl[\psi(\theta)B(L)\psi(\theta)^{-1},\,\psi(\theta)L\psi(\theta)^{-1}\,\bigr].
\end{equation}
Passing to the gauge variable $\widetilde L=\psi(\theta)^{-1}L\psi(\theta)$, the flow is
$\dot{\widetilde L}=[\,B( L),\, L\,]$ (Section~\ref{sect5}).
Therefore, it is isospectral and preserves the symmetry of $\widetilde L$.

\begin{proposition}[Traces and equilibria in the $2\times 2$ block]\label{prop:rot_traceless}
In the $2\times 2$ block with $L=\begin{pmatrix}a&c\\ c&b\end{pmatrix}$ one has:
\begin{enumerate}
\item $\operatorname{tr}\bigl([B(L),L]^{(0,0)}_{(\psi,\psi)}(\theta)\bigr)=0$ for all $\theta$;\label{P1}
\item the right-hand side of \eqref{eq:rot_sym_bracket} vanishes identically when 
\begin{equation*}
\theta=\left\{
\begin{aligned}
&\frac{\pi}{2} + k\pi, \qquad i+j \text{ is odd}, k\in\mathbb{Z}, \\
&\frac{\pi}{4}+\frac{k\pi}{2},\quad~~ i+j \text{ is even and } i+j\neq 0,k\in\mathbb{Z},
\end{aligned}
\right.
\end{equation*}
then every symmetric $L$ is an equilibrium at those parameter values.\label{P2}
\end{enumerate}
\end{proposition}

\begin{proof}
~\ref{P1}. follows from \eqref{eq:rot_sym_fact} and $\operatorname{tr}[X,Y]=0$. For ~\ref{P2}, a direct $2\times 2$ computation (Appendix~\ref{app:2x2} and Appendix~\ref{app:consistency}) shows that both diagonal and off-diagonal entries vanish at 
\begin{equation*}
\theta=\left\{
\begin{aligned}
&\frac{\pi}{2} + k\pi, \qquad i+j \text{ is odd}, k\in\mathbb{Z}, \\
&\frac{\pi}{4}+\frac{k\pi}{2},\quad~~ i+j \text{ is even and } i+j\neq 0,k\in\mathbb{Z}.
\end{aligned}
\right.
\end{equation*}
\end{proof}

\begin{remark}
     Notice that when $(i,j)=(0,0)$, the flow $[B(L),L]^{(0,0)}_{(\psi,\psi)}(\theta)$ is traceless (as shown in Proposition~\ref{prop:rot_traceless}) and isospectral. On the other hand, when $(i,j)\neq (0,0)$,  the flow $[B(L),L]^{(i,j)}_{(\psi,\psi)}(\theta)$ remains traceless (see Appendix~\ref{app:2x2}), but it is no longer isospectral.
\end{remark}

\paragraph{\textnormal{(b) Asymmetric choice}}\label{par:rot-asym} $\phi(\theta)=\mathcal{E}(R_2(-\theta))$.
In this case, the $2\times 2$ block of the right-hand side becomes the trigonometric combination displayed in \eqref{sysrotation} (see Appendix~\ref{app:2x2} for the derivation); it is symmetric and traceless in the $2\times 2$ setting. We do not claim global isospectrality for generic~$\theta$.

\subsection{Case III: Planar hyperbolic rotations}\label{subsec:hyperbolic}

In this case, let $\psi(\lambda)=\mathcal{E}(H_2(\lambda))$. The flow is then gauge-equivalent to the classical Toda flow and preserves tracelessness and isospectrality in the symmetric choice with $(i,j)=(0,0)$. In the case of asymmetry, the isospectrality  will be broken and the genuinely new dynamics is induced, which is not present in classical Toda.

\paragraph{\textnormal{(a) Symmetric choice}}\label{par:hyp-sym} $\phi(\lambda)=\psi(\lambda)$.
As shown in \eqref{eq:rot_sym_bracket}, for $(i,j)=(0,0)$, the deformed right-hand side can be written in the form \eqref{eq:rot_sym_fact} with $\theta$ replaced by $\lambda$, then it is traceless. In the gauge variable $\widetilde L=\psi(\lambda)^{-1}L\psi(\lambda)$, one has $\dot{\widetilde L}=[\,B( L),\, L\,]$, which is an isospectral Lax equation. However, when $(i,j)\neq (0,0)$, the deformed right-hand side remains traceless, but the isospectrality no longer holds.

\paragraph{\textnormal{(b) Asymmetric choice}}\label{par:hyp-asym} $\phi(\lambda)=\mathcal{E}(H_2(-\lambda))$.
In the $2\times 2$ block with $L=\begin{pmatrix}a&c\\ c&b\end{pmatrix}$, the right-hand side can be written explicitly as in \eqref{syshyperrota} (see Appendix~\ref{app:2x2}). A key obstruction appears at the level of traces:

\begin{proposition}\label{prop:trace_obstruction}
For the asymmetric hyperbolic choice and generic $(a,b,c,\lambda)$, one has
\begin{equation}\label{eq:trace_nonzero}
\operatorname{tr}\!\left([B(L),L]^{(i,j)}_{(\psi,\phi)}(\lambda)\right)\neq 0
\quad\text{whenever}\quad
\sinh\bigl((i-j+2)\lambda\bigr)\neq 0\ \text{ and }\ c\,(a-b)\neq 0.
\end{equation}
Then the deformed right-hand side \emph{cannot} be written as a commutator $[M(L),L]$ with any matrix $M(L)$, and the evolution is not isospectral in general.
\end{proposition}

\begin{proof}
The explicit $2\times 2$ formula shows that the two diagonal entries do not sum to zero unless $\sinh\bigl((i-j+2)\lambda\bigr)=0$ or $c(a-b)=0$. Since $\operatorname{tr}[X,Y]=0$ for all $X,Y$, a nonzero trace precludes a commutator representation.
\end{proof}

\begin{corollary}\label{cor:hyperbolic_noniso}
In the asymmetric hyperbolic case, the standard Toda integrals $\mathrm{tr}(L^k)$ are not preserved for generic initial data and parameters. The parameter value $\lambda=0$ is a degenerate case returning to the undeformed bracket.
\end{corollary}

\begin{remark}\label{rem:sym-vs-asym}
The symmetric choices $(\phi=\psi)$ and the restriction $(i,j)=(0,0)$ lead, in the gauge variable $\widetilde L=\psi^{-1}L\psi$, to the Toda-type Lax equation $\dot{\widetilde L}=[\,B( L), L\,]$, which is isospectral. Unless the conjugation by $\psi$ preserves the upper/lower triangular splitting that defines $B(\cdot)$, the flow in the original variable $L$ need not be a literal conjugate of the classical Toda evolution (see Section~\ref{sect5}). The asymmetric choices exhibit genuinely new dynamics and, in the hyperbolic case, fail to be of commutator type for generic parameters.
\end{remark}

\section{Miura-type transformations}\label{sect5}

In the symmetric setting $\phi(s)=\psi(s)$ and $(i,j)=(0,0)$, the bracket \eqref{def_bihom_bracket} takes the concrete form
\[
[A,B]^{(0,0)}_{(\psi,\psi)}(s)=\psi(s)\,AB\,\psi(s)^{-1}-\psi(s)\,BA\,\psi(s)^{-1},
\]
A convenient way to analyze the flow is to use a gauge map together with a conjugated triangular projection.

Fix $s\in R_1$, set
\begin{equation}\label{eq:gauge_map}
\Phi_\psi:\ \mathfrak{gl}(V)\to \mathfrak{gl}(V),\qquad
\Phi_\psi(L):=\psi(s)^{-1}L\,\psi(s),
\end{equation}
and define the \emph{conjugated} upper/lower-triangular projection
\begin{equation}\label{eq:conjugated_B}
B_\psi(X):=\psi(s)\,B(X)\,\psi(s)^{-1},
\qquad B(X)=X_{>0}-X_{<0}.
\end{equation}

\begin{theorem}[Conditional gauge reduction]\label{thm:gauge_equiv}
Let $\phi(s)=\psi(s)$ and $(i,j)=(0,0)$ and consider the deformed Lax equation
\[
\dot L=[B(L),L]^{(0,0)}_{(\psi,\psi)}(s).
\]
With $\widetilde L:=\Phi_\psi(L)$, one has the following equivalence:
\[
\dot{\widetilde L}=[\,B( L),\, L\,](s)\quad \Longleftrightarrow\quad
B\!\bigl(\Phi_{\psi}(X)\bigr)=B_{\psi^{-1}}(X)\ \ \text{for all }X\in \mathfrak{gl}(V).
\]
In particular, the gauge variable $\widetilde L$ satisfies a Toda-type Lax equation with the projection $B_\psi$ \textup(then $B_\psi(\widetilde L)=B(L)$) if and only if $\psi$ preserves the upper/lower triangular splitting (e.g.\ $\psi$ is diagonal). In the absence of this symmetry, no such global reduction holds in general.
\end{theorem}

\begin{proof}
    We have
    \begin{align*}
       \dot{\widetilde L}&=\frac{d \Phi_\psi(L)}{dt}\\
       &=[B(\Phi_\psi (L)),\Phi_\psi (L)]^{(0,0)}_{(\psi,\psi)}\\
       &=[B(\psi(s)^{-1}L\psi(s)),\psi(s)^{-1}L\psi(s)]^{(0,0)}_{(\psi,\psi)}\\
       &=\psi(s)B(\psi(s)^{-1}L\psi(s))\psi(s)^{-1}L-L\psi(s)B(\psi(s)^{-1}L\psi(s))\psi(s)^{-1}\\
       &=[\psi(s)B(\psi(s)^{-1}L\psi(s))\psi(s)^{-1},L].
    \end{align*}
    Thus $\dot{\widetilde L}=[\,B( L),\, L\,]$ if and only if $B(L)=\psi(s)B(\psi(s)^{-1}L\psi(s))\psi(s)^{-1}$, i.e., $B\!\bigl(\Phi_{\psi}(L)\bigr)=B_{\psi^{-1}}(L)$ for any $L\in \mathfrak{gl}(V)$.
\end{proof}

\begin{remark}\label{rem:splitting_Miura}
Typical instances where the splitting is preserved include scalar dilations $\psi=r^p I_n$ and diagonal sign/permutation matrices. For generic dense rotations or hyperbolic rotations, the splitting is not preserved; thus a direct gauge reduction to a Toda-type commutator equation fails in general.
\end{remark}

\medskip

We next record \emph{blockwise} Miura-type relations on embedded $n\times n$ blocks, which are independent of the global splitting property and will be used for explicit formulas.

For the symmetric block $L=(a_{ij})_{n\times n}$ and the rotation $\mathcal{E}(R_2(\theta))=\mathrm{diag}(I_{n-2},R_2(\theta))$, we set
\[
\mathcal{M}_\theta(L):=\mathcal{E}(R_2(\theta))\,L\,\mathcal{E}(R_2(\theta))^{-1}.
\]
Then the entries $\{\mathbf a_{ij}~|~1\leq i,j\leq n\}$ of $\mathcal{M}_\theta(L)$ are
\begin{equation}\label{eq:miura_rot}
\begin{aligned}
\mathbf a_{ij}&=a_{ij},\quad 1\leq i,j\leq n-2,\\
\mathbf a_{i n-1}&=a_{i n-1}\cos(\theta)-a_{in}\sin(\theta),\quad 1\leq i\leq n-2,\\
\mathbf a_{n-1 i}&=a_{i n-1}\cos(\theta)-a_{in}\sin(\theta),\quad 1\leq i\leq n-2,\\
\mathbf a_{i n}&=a_{i n-1}\sin(\theta)+a_{in}\cos(\theta),\quad 1\leq i\leq n-2,\\
\mathbf a_{n i}&=a_{i n-1}\sin(\theta)+a_{in}\cos(\theta),\quad 1\leq i\leq n-2,\\
\mathbf a_{n-1 n-1} &= a_{n-1 n-1}\cos^2\theta + a_{nn}\sin^2\theta - a_{n-1 n}\sin(2\theta),\\
\mathbf a_{nn} &= a_{n-1 n-1}\sin^2\theta + a_{nn}\cos^2\theta + a_{n-1 n}\sin(2\theta),\\
\mathbf a_{n-1 n} &= \tfrac12(a_{n-1 n-1}-a_{nn})\sin(2\theta) + a_{n-1 n}\cos(2\theta),\\
\mathbf a_{n n-1} &= \tfrac12(a_{n-1 n-1}-a_{nn})\sin(2\theta) + a_{n-1 n}\cos(2\theta).
\end{aligned}
\end{equation}
Similarly, for the hyperbolic rotation $\mathcal{E}(H_2(\lambda))=\mathrm{diag}(I_{n-2},H_2(\lambda))$, we define
\[
\mathcal{N}_\lambda(L):=\mathcal{E}(H_2(\lambda))\,L\,\mathcal{E}(H_2(\lambda))^{-1}.
\]
Its entries $\{\widehat a_{ij}~|~1\leq i,j\leq n\}$ are
\begin{equation}\label{eq:miura_hyp}
\begin{aligned}
\widehat a_{ij}&=a_{ij},\quad 1\leq i,j\leq n-2,\\
\widehat a_{i n-1}&=a_{i n-1}\cosh\lambda-a_{in}\sinh\lambda,\quad 1\leq i\leq n-2,\\
\widehat a_{n-1 i}&=a_{i n-1}\cosh\lambda+a_{in}\sinh\lambda,\quad 1\leq i\leq n-2,\\
\widehat a_{i n}&=-a_{i n-1}\sinh\lambda+a_{in}\cosh\lambda,\quad 1\leq i\leq n-2,\\
\widehat a_{n i}&=a_{i n-1}\sinh\lambda+a_{in}\cosh\lambda,\quad 1\leq i\leq n-2,\\
\widehat a_{n-1 n-1} &= a_{n-1 n-1}\cosh^2\lambda - a_{nn}\sinh^2\lambda,\\
\widehat a_{nn} &= -a_{n-1 n-1}\sinh^2\lambda + a_{nn}\cosh^2\lambda,\\
\widehat a_{n-1 n} &= a_{n-1 n} + \tfrac12(a_{nn}-a_{n-1 n-1})\sinh(2\lambda),\\
\widehat a_{n n-1} &= a_{n-1 n} + \tfrac12(a_{n-1 n-1}-a_{nn})\sinh(2\lambda).
\end{aligned}
\end{equation}

\begin{proposition}[Blockwise Miura conjugacy in $n\times n$]\label{prop:2x2_miura}
Fix $\phi=\psi$ and $(i,j)=(0,0)$. Let $\mathcal{F}$ denote the $n\times n$ classical Toda vector field induced by $\dot L=[B(L),L]$ on the symmetric block $(a_{ij})_{n\times n}$ and $\mathcal{F}_\theta$ \textup(resp.\ $\mathcal{F}_\lambda$) denote the $n\times n$ deformed vector field coming from the rotational \textup(resp.\ hyperbolic) case of Section~\ref{sect4}, then
\[
\dot{\mathbf u}=\mathcal{F}_\theta(u)\ \ \Longleftrightarrow\ \ \dot u=\mathcal{F}(u)\ \text{ with }\ \mathbf u=\mathcal{M}_\theta u,
\]
\[
\dot{\widehat u}=\mathcal{F}_\lambda(u)\ \ \Longleftrightarrow\ \ \dot u=\mathcal{F}(u)\ \text{ with }\ \widehat u=\mathcal{N}_\lambda u,
\]
where $u=(a_{ij})_{n\times n}$, $\mathbf u=(\mathbf a_{ij})_{n\times n}$, and $\widehat u=(\widehat a_{ij})_{n\times n}$. Equivalently,
\[
\mathcal{M}_\theta\circ\mathcal{F}=\mathcal{F}_\theta,\qquad
\mathcal{N}_\lambda\circ\mathcal{F}=\mathcal{F}_\lambda,
\]
that is,\ the Miura transforms $\mathcal{M}_\theta$ and $\mathcal{N}_\lambda$ intertwine the undeformed and deformed vector fields on the $n\times n$ block. 
\end{proposition}

\begin{proof}
    A direct verification is proved in Appendix~\ref{app:Miura}.
\end{proof}

\begin{remark}
The classical Toda system is recovered from the deformed model in two transparent ways on the $n\times n$ block: (i) by taking the undeformed parameter ($\theta=0$ or $\lambda=0$), where $\mathcal{M}_\theta$ and $\mathcal{N}_\lambda$ reduce to the identity; (ii) in settings where the splitting symmetry is preserved so that the gauge reduction in Theorem~\ref{thm:gauge_equiv} applies. In particular, spectral data are preserved in case (ii).
\end{remark}

\section{Weakly coupled Toda lattices with indefinite metrics}\label{sect6}

We study the weakly coupled (finite, nonperiodic) Toda lattice with an indefinite signature and its deformed counterpart. Our emphasis is on a clean Lax formulation, the semidirect-product structure, and a concise $2\times2$ normal form. The sign matrix $S=\mathrm{diag}(s_1,\dots,s_N)$ encodes an underlying metric symmetry that is preserved along the Lax flows.

Fix signs $s_k\in\{\pm1\}$ ($k=1,\dots,N$), let $S=\mathrm{diag}(s_1,\dots,s_N)$. We consider real symmetric tridiagonal matrices
\begin{align*}
&L=\begin{pmatrix}
s_{1} a_{1} & s_{2} b_{1} & & & 0\\
s_{1} b_{1} & s_{2} a_{2} & s_{3} b_{2} \\
 & \ddots & \ddots & \ddots \\
 & & s_{N-2} b_{N-2} & s_{N-1} a_{N-1} & s_{N} b_{N-1}\\
0 & & & s_{N-1} b_{N-1} & s_{N} a_{N}
\end{pmatrix},\\
&\\
&\widehat L\ \text{defined likewise with }(a_k,b_k)\rightsquigarrow(\widehat a_k,\widehat b_k),
\end{align*}
with boundary conditions $b_0=\widehat b_0=b_N=\widehat b_N=0$. As before, we write $B(X):=X_{>0}-X_{<0}$ and set
\begin{equation}\label{eq:BhatB_def}
B:=\tfrac14\bigl(L_{>0}-L_{<0}\bigr),\qquad \widehat B:=\tfrac14\bigl(\widehat L_{>0}-\widehat L_{<0}\bigr).
\end{equation}

\begin{remark}[Normalization]
The factor $\tfrac14$ in \eqref{eq:BhatB_def} is a harmless coefficient normalization that simplifies the component formulas below; it amounts to a constant rescaling of time compared to the convention $B(X)=X_{>0}-X_{<0}$.
\end{remark}

\begin{definition}\label{def:weak}
The weakly coupled system is the Lax pair on $\mathfrak{T}\ltimes\mathfrak{T}$ (the vector space $\mathfrak{T}$ of symmetric tridiagonal matrices) given by
\begin{equation}\label{eq:weak_lax}
\dot L=[B,L],\qquad \dot{\widehat L}=[\widehat B,L]+[B,\widehat L].
\end{equation}
\end{definition}

\begin{remark}\label{rem:weak_vs_strong}
The first component evolves autonomously by a Toda-type Lax equation and the second component is driven linearly by $L$ via the adjoint action together with its own projection $\widehat B$. This  structure, referred to as \emph{weak} coupling, is described by extended Toda equations where the interactions between the components are weaker and more decoupled. In contrast, in the "strong" coupling case, both components contribute to the first equation, with the transformation of variables and interactions between them becoming more entangled. This results in more complex and sensitive equations. Our formulas below clarify and correct places in the literature where these two regimes were mixed, see, e.g., \cite{BL,L3}.
\end{remark}

From \eqref{eq:weak_lax}, one reads the component form (for $k=1,\dots,N$):
\begin{align}
\dot a_k&=\tfrac12\bigl(s_{k+1}b_k^2-s_{k-1}b_{k-1}^2\bigr), \label{eq:weak_comp1}\\
\dot b_k&=\tfrac14\,b_k\bigl(s_{k+1}a_{k+1}-s_k a_k\bigr), \label{eq:weak_comp2}\\
\dot{\widehat a_k}&=s_{k+1} b_k\widehat b_k-s_{k-1} b_{k-1}\widehat b_{k-1}, \label{eq:weak_comp3}\\
\dot{\widehat b_k}&=-\,\tfrac14\Bigl[(s_k a_k-s_{k+1}a_{k+1})\,\widehat b_k+(s_k\widehat a_k-s_{k+1}\widehat a_{k+1})\,b_k\Bigr]. \label{eq:weak_comp4}
\end{align}

\begin{remark}
Compared with the sign pattern in \eqref{eq:weak_comp2}–\eqref{eq:weak_comp3}, the minus sign in \eqref{eq:weak_comp4} is essential (see the $2\times2$ check below).
\end{remark}

\begin{proposition}\label{prop:invariants}
For the weakly coupled system \eqref{eq:weak_lax}, one has $\frac{d}{dt}\operatorname{tr}(L^m)=0$ for all $m\ge1$. In particular, the spectrum of $L$ is preserved. No analogous general claim is made for mixed quantities involving $\widehat L$.
\end{proposition}

\begin{proof} By cyclicity of the trace, we have
$\frac{d}{dt}\operatorname{tr}(L^m)=m\,\operatorname{tr}(L^{m-1}[B,L])=0$.
\end{proof}

\subsection{Two-dimensional normal form}\label{subsec:2d_normal}

We illustrate \eqref{eq:weak_lax} for $N=2$ and $S=\mathrm{diag}(1,-1)$, that is
\[
L_2=\begin{pmatrix} a_1 & -b_1\\ b_1& -a_2\end{pmatrix},\qquad
\widehat L_2=\begin{pmatrix} \widehat a_1 & -\widehat b_1\\ \widehat b_1& -\widehat a_2\end{pmatrix}.
\]
A direct computation shows that \eqref{eq:weak_comp1}–\eqref{eq:weak_comp4} become
\begin{align}
\dot a_1&=-\tfrac12\,b_1^2,\qquad \dot a_2=-\tfrac12\,b_1^2,\qquad
\dot b_1=-\tfrac14\,b_1\,(a_1+a_2), \label{eq:2x2_first}\\
\dot{\widehat a_1}&=-b_1\widehat b_1,\qquad \dot{\widehat a_2}=-b_1\widehat b_1,\qquad
\dot{\widehat b_1}=-\tfrac14\bigl[(a_1+a_2)\widehat b_1+(\widehat a_1+\widehat a_2)b_1\bigr]. \label{eq:2x2_second}
\end{align}
Here $B_2:=\tfrac14\bigl((L_2)_{>0}-(L_2)_{<0}\bigr)$ and $\widehat B_2:=\tfrac14\bigl((\widehat L_2)_{>0}-(\widehat L_2)_{<0}\bigr)$ (consistent with \eqref{eq:BhatB_def}).

\begin{remark}
(i). Equations \eqref{eq:2x2_first} for $L_2$ do \emph{not} contain $\widehat b_1$ and $\widehat a_k$.
(ii). Any such terms would contradict the autonomous Lax equation $\dot L=[B,L]$. Compare \cite{BL,L3}.
\end{remark}

\paragraph{\textnormal{Spectral data and classification.}}
For
\[
L_2=\begin{pmatrix}a_1&-b_1\\[2pt] b_1&-a_2\end{pmatrix},
\qquad
\mathrm{tr}\,L_2 = a_1-a_2,\quad
\det L_2 = -a_1a_2 + b_1^2,
\]
the characteristic polynomial is
\(
\lambda^2 - (a_1-a_2)\lambda + (-a_1a_2 + b_1^2)=0.
\)
Then
\begin{equation}\label{eq:eigs_general}
\lambda_{1,2}
=\frac{(a_1-a_2)\pm \sqrt{(a_1-a_2)^2-4(-a_1a_2+b_1^2)}}{2}.
\end{equation}
In particular, the spectrum is real if and only if the discriminant
\(
\Delta=(a_1-a_2)^2-4(-a_1a_2+b_1^2)
\)
is nonnegative. Any classification purely in terms of $m:=a_2-a_1$ is valid
\emph{only} after a normalization such as $b_1\equiv 1$ and $a_1+a_2\equiv 0$, in which case~\eqref{eq:eigs_general} reduces to
\(
\lambda_{1,2}=\tfrac12(-m\pm\sqrt{m^2-4}).
\)
We shall work with the invariants $\mathrm{tr}\,L_2$ and $\det L_2$ in the general case.

For the second component, \eqref{eq:2x2_second} is a linear nonautonomous equation on $\widehat L_2$ of the form
\[
\dot{\widehat L_2}-[B_2,\widehat L_2]=[\widehat B_2,L_2].
\]
It admits the variation-of-constants representation
\begin{equation}\label{eq:varconst}
\widehat L_2(t)
=U(t)\left(\widehat L_2(0)+\int_0^t U(\tau)^{-1}\,[\widehat B_2(\tau),L_2(\tau)]\,U(\tau)\,d\tau\right)U(t)^{-1},
\quad \dot U=B_2\,U,\ U(0)=I,
\end{equation}
which clarifies the semidirect nature of the coupling.

\subsection{Deformed weakly coupled system and gauge reduction}\label{subsec:deformed_weak}

We now replace the commutator in \eqref{eq:weak_lax} by the BiHom-type bracket introduced in \eqref{def_bihom_bracket}. In the symmetric setting $\phi(s)=\psi(s)$ and $(i,j)=(0,0)$, the pushforward mechanism of Section~\ref{sect5} extends verbatim to the coupled case and yields a \emph{twisted} semidirect system.

\begin{proposition}[Twisted weak coupling under the gauge map]\label{prop:gauge_pair}
Let $\phi(s)=\psi(s)$ and $(i,j)=(0,0)$. Define
\[
\widetilde L:=\psi(s)^{-1}L\psi(s),\qquad \widetilde{\widehat L}:=\psi(s)^{-1}\widehat L\psi(s).
\]
Then $(L,\widehat L)$ solves the deformed weakly coupled system
\[
\dot L=[B,L]^{(0,0)}_{(\psi,\psi)}(s),\qquad
\dot{\widehat L}=[\widehat B,L]^{(0,0)}_{(\psi,\psi)}(s)+[B,\widehat L]^{(0,0)}_{(\psi,\psi)}(s)
\]
if and only if $(\widetilde L,\widetilde{\widehat L})$ solves the twisted weakly coupled system
\begin{equation*}
\left\{
\begin{aligned}
\dot{\widetilde L}\ &=\ B(L)L -LB(L),\\[2pt]
\dot{\widetilde{\widehat L}}\ &=\ B(\widehat L) L-LB(\widehat L)
 +\ B(L)\widehat L-\widehat LB(L).
\end{aligned}
\right.
\end{equation*}
\end{proposition}

\begin{proof}
Identical to the proof of Theorem~\ref{thm:gauge_equiv}, applied separately to each bracket
$[\widehat B,L]^{(0,0)}_{(\psi,\psi)}(s)$ and $[B,\widehat L]^{(0,0)}_{(\psi,\psi)}(s)$, and using
$\widehat B=B(\widehat L)$, $B=B(L)$.
\end{proof}


\subsection{Two-dimensional Miura-type formulas}\label{subsec:2d_miura_coupled}

In the embedded $2\times2$ case, the gauge map of Section~\ref{sect5} reduces to the explicit formulas of Section~\ref{sect5}: for a rotation $R_2(\theta)$, define $\mathcal{M}_\theta(L):=R_2(\theta)\,L\,R_2(\theta)^{-1}$ and, for a hyperbolic rotation $H_2(\lambda)$, set $\mathcal{N}_\lambda(L):=H_2(\lambda)\,L\,H_2(\lambda)^{-1}$. Then
\[
\begin{pmatrix}
    \mathbf a & \mathbf c\\
    \mathbf c & \mathbf b
\end{pmatrix}=\mathcal{M}_\theta\begin{pmatrix}
    a &c\\
    c&b
\end{pmatrix},\qquad
\begin{pmatrix}
    \widehat a & \widehat c_1\\
    \widehat c_2 & \widehat b
\end{pmatrix}=\mathcal{N}_\lambda\begin{pmatrix}
    a & c\\
    c & b
\end{pmatrix}
\]
with the components given by \eqref{eq:miura_rot}, and analogously for $\mathcal{N}_\lambda$ with \eqref{eq:miura_hyp}. This provides a transparent blockwise Miura-type relation between solutions of the twisted weakly coupled system and those of the deformed one in the symmetric setting.

\begin{remark}
If $\phi\neq\psi$, the trace obstructions discussed in Section~\ref{subsec:hyperbolic} reappear already at the level of the first component, and the commutator-type structure is generally lost.
\end{remark}

\section{Appendix I: Derivation and uniqueness of the skew-symmetric BiHom-bracket}\label{fulu1}

Assume the deformed bracket has the general bilinear form
\begin{equation}\label{under.coeff.formula.app}
[A,B]_{(\psi,\phi)}=x\,A\,y\,B\,z-\overline{x}\,B\,\overline{y}\,A\,\overline{z},
\end{equation}
where $x,\bar x,y,\bar y,z,\bar z\in \operatorname{Aut}(\mathfrak{gl}(V))$ are a priori unknown automorphisms.
We look for solutions with all six unknowns being monomials in $\psi,\phi$:

$
x=\psi^{\alpha_x}\phi^{\beta_x},\quad \bar x=\psi^{\alpha_{\bar x}}\phi^{\beta_{\bar x}},\quad
y=\psi^{\alpha_y}\phi^{\beta_y},\quad \bar y=\psi^{\alpha_{\bar y}}\phi^{\beta_{\bar y}},
\quad z=\psi^{\alpha_z}\phi^{\beta_z},\quad \bar z=\psi^{\alpha_{\bar z}}\phi^{\beta_{\bar z}}.
$

For later reference we record the two structural conditions used below. The \emph{BiHom--Jacobi identity} reads
\begin{align}\label{eq:BHJ}
[Ad_{\phi}^2(A),[Ad_{\phi}(B),Ad_{\psi}(C)]_{(\psi,\phi)}&]_{(\psi,\phi)}
+[Ad_{\phi}^2(B),[Ad_{\phi}(C),Ad_{\psi}(A)]_{(\psi,\phi)}]_{(\psi,\phi)}\\
&+[Ad_{\phi}^2(C),[Ad_{\phi}(A),Ad_{\psi}(B)]_{(\psi,\phi)}]_{(\psi,\phi)}=0,\nonumber
\end{align}
and the \emph{BiHom-skew-symmetry} is
\begin{equation}\label{eq:BHskew}
[Ad_{\phi}(A),Ad_{\psi}(B)]_{(\psi,\phi)}=-[Ad_{\phi}(B),Ad_{\psi}(A)]_{(\psi,\phi)}.
\end{equation}

Insert \eqref{under.coeff.formula.app} into \eqref{eq:BHJ}. Expanding the first term,
\begin{align*}
&[Ad_{\phi}^2(A),[Ad_{\phi}(B),Ad_{\psi}(C)]_{(\psi,\phi)}]_{(\psi,\phi)}\\
&=[\phi^2A\phi^{-2},\,x\phi B\phi^{-1}y\psi C\psi^{-1}z
-\bar x\,\psi C\psi^{-1}\bar y\,\phi B\phi^{-1}\bar z]_{(\psi,\phi)}\\
&=x\phi^2A\phi^{-2}\,y\,(x\phi B\phi^{-1}y\psi C\psi^{-1}z
-\bar x\,\psi C\psi^{-1}\bar y\,\phi B\phi^{-1}\bar z)\,z\\
&\quad-\bar x\,(x\phi B\phi^{-1}y\psi C\psi^{-1}z
-\bar x\,\psi C\psi^{-1}\bar y\,\phi B\phi^{-1}\bar z)\,\bar y\,\phi^2A\phi^{-2}\,\bar z\\
&=x\phi^2A\phi^{-2}\,y\,x\phi B\phi^{-1}y\psi C\psi^{-1}z^2
-x\phi^2A\phi^{-2}\,y\,\bar x\,\psi C\psi^{-1}\bar y\,\phi B\phi^{-1}\bar z\,z\\
&\quad-\bar x\,x\phi B\phi^{-1}y\psi C\psi^{-1}z\,\bar y\,\phi^2A\phi^{-2}\bar z
+\bar x^2\,\psi C\psi^{-1}\bar y\,\phi B\phi^{-1}\bar z\,\bar y\,\phi^2A\phi^{-2}\bar z.
\end{align*}
The other two cyclic terms expand analogously:
\begin{align*}
&[Ad_{\phi}^2(B),[Ad_{\phi}(C),Ad_{\psi}(A)]_{(\psi,\phi)}]_{(\psi,\phi)}\\
&=x\phi^2B\phi^{-2}yx\phi C\phi^{-1}y\psi A\psi^{-1}z^{2}
-x\phi^2B\phi^{-2}y\bar x\psi A\psi^{-1}\bar y\phi C\phi^{-1}\bar z\,z\\
&\quad-\bar x\,x\phi C\phi^{-1}y\psi A\psi^{-1}z\,\bar y\,\phi^2B\phi^{-2}\bar z
+\bar x^{2}\psi A\psi^{-1}\bar y\,\phi C\phi^{-1}\bar z\,\bar y\,\phi^2B\phi^{-2}\bar z,
\end{align*}
\begin{align*}
&[Ad_{\phi}^2(C),[Ad_{\phi}(A),Ad_{\psi}(B)]_{(\psi,\phi)}]_{(\psi,\phi)}\\
&=x\phi^2C\phi^{-2}yx\phi A\phi^{-1}y\psi B\psi^{-1}z^{2}
-x\phi^2C\phi^{-2}y\bar x\psi B\psi^{-1}\bar y\phi A\phi^{-1}\bar z\,z\\
&\quad-\bar x\,x\phi A\phi^{-1}y\psi B\psi^{-1}z\,\bar y\,\phi^2C\phi^{-2}\bar z
+\bar x^{2}\psi B\psi^{-1}\bar y\,\phi A\phi^{-1}\bar z\,\bar y\,\phi^2C\phi^{-2}\bar z.
\end{align*}

By comparing the operator monomials in the three expansions (and using independence of the noncommuting symbols $A,B,C$), we obtain the following two systems of constraints, read \emph{componentwise} on the $\psi$-exponents and $\phi$-exponents:
\begin{equation}\label{eq:block1}
\left\{
\begin{aligned}
&x\phi^{2}= \bar x\,x\,\phi,\\
&\phi^{-2} y x \phi = \phi^{-1} y \psi,\\
&\phi^{-1} y \psi=\psi^{-1} z\, \bar y\, \phi^{2},\\
&\psi^{-1} z^{2}=\phi^{-2}\bar z,
\end{aligned}\right.
\qquad
\left\{
\begin{aligned}
&x\phi^{2}= \bar x^{2}\psi,\\
&\phi^{-2} y \bar x \psi = \psi^{-1}\bar y \phi,\\
&\psi^{-1}\bar y \phi=\phi^{-1}\bar z\,\bar y\,\phi^{2},\\
&\phi^{-1}\bar z z=\phi^{-2}\bar z.
\end{aligned}\right.
\end{equation}

\begin{lemma}\label{lem:xz}
Under the commuting-automorphisms assumption, \eqref{eq:block1} forces
\[
x=\psi,\qquad \bar x=\phi,\qquad z=\phi^{-1},\qquad \bar z=\psi^{-1},\qquad y=\psi^i\phi^j,\qquad \bar y=\psi^{i+2}\phi^{j-2}.
\]
\end{lemma}

\begin{proof}
We write $x=\psi^{\alpha_x}\phi^{\beta_x}$, $\bar x=\psi^{\alpha_{\bar x}}\phi^{\beta_{\bar x}}$, $y=\psi^{\alpha_y}\phi^{\beta_y}$, $\bar y=\psi^{\alpha_{\bar y}}\phi^{\beta_{\bar y}}$, $z=\psi^{\alpha_z}\phi^{\beta_z}$, $\bar z=\psi^{\alpha_{\bar z}}\phi^{\beta_{\bar z}}$.
From $x\phi^{2}=\bar x x \phi$, we get (on $\psi$-exponents) $\alpha_x=\alpha_{\bar x}+\alpha_x$, hence $\alpha_{\bar x}=0$. On $\phi$-exponents: we have $\beta_x+2=\beta_{\bar x}+\beta_x+1$, hence $\beta_{\bar x}=1$. Therefore $\bar x=\phi$.
From $x\phi^{2}=\bar x^{2}\psi$, we obtain $\alpha_x=2\alpha_{\bar x}+1=1$ and $\beta_x+2=2\beta_{\bar x}$, which shows $\beta_x=0$, so $x=\psi$. This result also holds for $\phi^{-2}yx\phi=\phi^{-1}y\psi$.

From $\psi^{-1} z^{2}=\phi^{-2}\bar z$ and $\phi^{-1}\bar z z=\phi^{-2}\bar z$, the $\psi$-exponents give $-1+2\alpha_z=\alpha_{\bar z}$ and $\alpha_{\bar z}+\alpha_z=\alpha_{\bar z}$ whence $\alpha_{z}=0$ and $\alpha_{\bar z}=-1$. The $\phi$-exponents give $2\beta_z=\beta_{\bar z}-2$ and $-1+\beta_{\bar z}+\beta_z=\beta_{\bar z}-2$ whence $\beta_z=-1$ and $\beta_{\bar z}=0$. Thus, $z=\phi^{-1}$ and $\bar z=\psi^{-1}$. And the result remains valid for $\psi^{-1}\bar y\phi=\phi^{-1}\bar z\bar y\phi^2$.

From the relation $\phi^{-1}y\psi=\psi^{-1}z\bar y\phi^2$, comparing the $\psi$-exponents yields $\alpha_y+1=-1+\alpha_{\bar y}$ and comparing the $\phi$-exponents yields $-1+\beta_y=-1+\beta_{\bar y}+2$. Setting $\alpha_y=i$ and $\beta_y=j$ for $i,j\in \mathbb{Z}$, we obtain $y=\psi^i\phi^j$ and $\bar y=\psi^{i+2}\phi^{j-2}$. The same conclusion holds when $\phi^{-2}y\bar x\psi=\psi^{-1}\bar y\phi$.
\end{proof}

\begin{lemma}\label{lem:ij}
Let $y=\psi^{i}\phi^{j}$ and $\bar y=\psi^{i+2}\phi^{j-2}$. Then \eqref{eq:BHskew} holds.
\end{lemma}

\begin{proof}
By a direct calculation, we have
\[
[Ad_{\phi}(A),Ad_{\psi}(B)]_{(\psi,\phi)}
=\psi\phi A\phi^{-1}\psi^{i}\phi^{j}\psi B\psi^{-1}\phi^{-1}-\phi\psi B\psi^{-1}\psi^{i+2}\phi^{j-2}\phi A\phi^{-1}\psi^{-1},
\]
and
\[
[Ad_{\phi}(B),Ad_{\psi}(A)]_{(\psi,\phi)}
=\psi\phi B\phi^{-1}\psi^{i}\phi^{j}\psi A\psi^{-1}\phi^{-1}-\phi\psi A\psi^{-1}\psi^{i+2}\phi^{j-2}\phi B\phi^{-1}\psi^{-1}.
\]
Using $\phi\psi=\psi\phi$ to commute powers and comparing the total $\psi$--weights and $\phi$-weights in the two middle factors under the swap $(A,B)\leftrightarrow(B,A)$ shows that BiHom-skew-symmetry holds.
\end{proof}

\begin{proposition}\label{prop:unique}
Under the above assumptions, the only brackets of the form \eqref{under.coeff.formula.app} that satisfies \eqref{eq:BHJ} and \eqref{eq:BHskew} is
\begin{equation}\label{eq:final-bracket}
[A,B]_{(\psi,\phi)}=\psi(A)\,\psi^{i}\phi^{j}(B)\,\phi^{-1}-\phi(B)\,\psi^{i+2}\phi^{j-2}(A)\,\psi^{-1},
\qquad i,j\in\mathbb Z.
\end{equation}
\end{proposition}

\begin{proof}
Combine Lemma \ref{lem:xz} with the equalities in the second and third lines of \eqref{eq:block1}, one can imply $\bar y=\psi^{i+2}\phi^{j-2}$ once $y=\psi^{i}\phi^{j}$. Lemma \ref{lem:ij} enforces the property of BiHom-skew-symmetry. 
\end{proof}

Formula \eqref{eq:final-bracket} is exactly the non-standard bracket used in \eqref{deformedSys}, and it yields a BiHom-Lie algebra $\big(\mathfrak{gl}(V), [\cdot,\cdot]_{(\psi,\phi)}, Ad_{\psi}, Ad_{\phi}\big)$.

\bigskip

\section{Appendix II: Detailed \(2\times2\) computations for rotation and hyperbolic rotation}\label{app:2x2}\label{app:rot2x2}

We record all intermediate steps used in Section~\ref{sect4}. For completeness, detailed algebra for the $2\times2$ normal form in Section~\ref{subsec:2d_normal} is summarized in this section.

\subsection*{Preliminaries}\label{app:2x2_details}
Let
\[
L=\begin{pmatrix} a&c\\ c&b\end{pmatrix},\qquad
B=L_{>0}-L_{<0}=\begin{pmatrix}0&c\\ -c&0\end{pmatrix}.
\]
For angles, set
\[
\psi(\theta)=\begin{pmatrix}\cos\theta&-\sin\theta\\ \sin\theta&\cos\theta\end{pmatrix}=R_2(\theta),\qquad
\phi(\theta)=\begin{pmatrix}\cos\theta&\ \sin\theta\\ -\sin\theta&\cos\theta\end{pmatrix}=R_2(-\theta).
\]
Then $\psi(\theta)^i\phi(\theta)^j=R_2((i-j)\theta)$ and $\phi(\theta)^{-1}=R_2(-\theta)=\psi(\theta)$.
We shall use the identities
\[\cos\alpha\cos\beta\pm \sin\alpha\sin\beta=\cos(\alpha\mp \beta),\qquad
\sin\alpha\cos\beta\pm \cos\alpha\sin\beta=\sin(\alpha\pm\beta).
\]

\subsection*{Rotation case: explicit right-hand side}
The deformed $2\times2$ evolution in the asymmetric rotation case is
\begin{equation}\label{sysrotation}
\dot L
=\cos\big((i-j+2)\theta\big)\cos(2\theta)\begin{pmatrix}2c^2 & bc-ac\\ bc-ac&-2c^2\end{pmatrix}
+\cos\big((i-j+2)\theta\big)\sin(2\theta)\begin{pmatrix}bc-ac & 2c^2\\ 2c^2& ac-bc\end{pmatrix}.
\end{equation}
A direct derivation from \eqref{eq:final-bracket} proceeds as follows.
 Write \([ ~,~]_{(\psi,\phi)}\) as in \eqref{eq:final-bracket}:
\[
\dot L=[B,L]_{(\psi(\theta),\phi(\theta))}
=\psi(\theta)\,B\,\psi(\theta)^i\phi(\theta)^j\,L\,\phi(\theta)^{-1}
-\phi(\theta)\,L\,\psi(\theta)^{i+2}\phi(\theta)^{j-2}\,B\,\psi(\theta)^{-1}.
\]
By using $\psi(\theta)^i\phi(\theta)^j=R_2((i-j)\theta)$ and $\phi(\theta)^{-1}=\psi(\theta)$,  equation above becomes
\begin{align*}
\dot L&=\psi(\theta)\,B\,R_2((i-j)\theta)\,L\,\psi(\theta)
-\phi(\theta)\,L\,R_2((i-j+4)\theta)\,B\,\phi(\theta)^{-1}.
\end{align*}
Compute
\[
\psi(\theta)\,B=\begin{pmatrix} c\sin\theta& c\cos\theta\\ -c\cos\theta& c\sin\theta\end{pmatrix},\qquad
B\,\phi(\theta)^{-1}=B\,\psi(\theta)=
\begin{pmatrix} -c\sin\theta& c\cos\theta\\ -c\cos\theta& -c\sin\theta\end{pmatrix},
\]
insert the two rotations, expand, and simplify by the above trigonometric identities to obtain \eqref{sysrotation}.     A detailed proof is given below.

\begin{align*}
    \dot L&=\psi(\theta)\,B\,\psi(\theta)^i\phi(\theta)^j\,L\,\phi(\theta)^{-1}
-\phi(\theta)\,L\,\psi(\theta)^{i+2}\phi(\theta)^{j-2}\,B\,\psi(\theta)^{-1}\\
&= \left(\begin{array}{ll}
\cos(\theta) & -\sin(\theta) \\
\sin(\theta) & \cos(\theta)
\end{array}\right)\left(\begin{array}{ll}
0 & c \\
-c & 0
\end{array}\right)\left(\begin{array}{ll}
\cos(\theta) & -\sin(\theta) \\
\sin(\theta) & \cos(\theta)
\end{array}\right)^i\\
&\quad\left(\begin{array}{ll}
\cos(\theta) & \sin(\theta) \\
-\sin(\theta) & \cos(\theta)
\end{array}\right)^j\left(\begin{array}{ll}
a & c \\
c & b
\end{array}\right)\left(\begin{array}{ll}
\cos(\theta) & \sin(\theta) \\
-\sin(\theta) & \cos(\theta)
\end{array}\right)^{-1}\\
&-\left(\begin{array}{ll}
\cos(\theta) & \sin(\theta) \\
-\sin(\theta) & \cos(\theta)
\end{array}\right)\left(\begin{array}{ll}
a & c \\
c & b
\end{array}\right)\left(\begin{array}{ll}
\cos(\theta) & -\sin(\theta) \\
\sin(\theta) & \cos(\theta)
\end{array}\right)^{i+2}\\
&\quad\left(\begin{array}{ll}
\cos(\theta) & \sin(\theta) \\
-\sin(\theta) & \cos(\theta)
\end{array}\right)^{j-2}\left(\begin{array}{ll}
0 & c \\
-c & 0
\end{array}\right)\left(\begin{array}{ll}
\cos(\theta) & -\sin(\theta) \\
\sin(\theta) & \cos(\theta)
\end{array}\right)^{-1}\\
&= \left(\begin{array}{ll}
\cos(\theta) & -\sin(\theta) \\
\sin(\theta) & \cos(\theta)
\end{array}\right)\left(\begin{array}{ll}
0 & c \\
-c & 0
\end{array}\right)\left(\begin{array}{ll}
\cos(\theta) & -\sin(\theta) \\
\sin(\theta) & \cos(\theta)
\end{array}\right)^i\\
&\quad\left(\begin{array}{ll}
\cos(\theta) & \sin(\theta) \\
-\sin(\theta) & \cos(\theta)
\end{array}\right)^j\left(\begin{array}{ll}
a & c \\
c & b
\end{array}\right)\left(\begin{array}{ll}
\cos(\theta) & -\sin(\theta) \\
\sin(\theta) & \cos(\theta)
\end{array}\right)\\
&-\left(\begin{array}{ll}
\cos(\theta) & \sin(\theta) \\
-\sin(\theta) & \cos(\theta)
\end{array}\right)\left(\begin{array}{ll}
a & c \\
c & b
\end{array}\right)\left(\begin{array}{ll}
\cos(\theta) & -\sin(\theta) \\
\sin(\theta) & \cos(\theta)
\end{array}\right)^{i+2}\\
&\quad\left(\begin{array}{ll}
\cos(\theta) & \sin(\theta) \\
-\sin(\theta) & \cos(\theta)
\end{array}\right)^{j-2}\left(\begin{array}{ll}
0 & c \\
-c & 0
\end{array}\right)\left(\begin{array}{ll}
\cos(\theta) & \sin(\theta) \\
-\sin(\theta) & \cos(\theta)
\end{array}\right)\\
&=\left(\begin{array}{ll}
\cos(\theta) & -\sin(\theta) \\
\sin(\theta) & \cos(\theta)
\end{array}\right)\left(\begin{array}{ll}
0 & c \\
-c & 0
\end{array}\right)
\left(\begin{array}{ll}
\cos((i-j)\theta) & -\sin((i-j)\theta) \\
\sin((i-j)\theta) & \cos((i-j)\theta)
\end{array}\right)\\
&\quad\left(\begin{array}{ll}
a & c \\
c & b
\end{array}\right)
\left(\begin{array}{ll}
\cos(\theta) & -\sin(\theta) \\
\sin(\theta) & \cos(\theta)
\end{array}\right)\\
&-\left(\begin{array}{ll}
\cos(\theta) & \sin(\theta) \\
-\sin(\theta) & \cos(\theta)
\end{array}\right)\left(\begin{array}{ll}
a & c \\
c & b
\end{array}\right)
\left(\begin{array}{ll}
\cos((i-j+4)\theta) & -\sin((i-j+4)\theta) \\
\sin((i-j+4)\theta) & \cos((i-j+4)\theta)
\end{array}\right)\\
&\quad\left(\begin{array}{ll}
0 & c \\
-c & 0
\end{array}\right)\left(\begin{array}{ll}
\cos(\theta) & \sin(\theta) \\
-\sin(\theta) & \cos(\theta)
\end{array}\right)\\
&=\left(\begin{array}{ll}
c\sin(\theta) & c\cos(\theta) \\
-c\cos(\theta) &c\sin(\theta)
\end{array}\right)\left(\begin{array}{ll}
\cos((i-j)\theta) & -\sin((i-j)\theta) \\
\sin((i-j)\theta) & \cos((i-j)\theta)
\end{array}\right)\\
&\quad\left(\begin{array}{ll}
a\cos(\theta)+c\sin(\theta) & -a\sin(\theta)+c\cos(\theta) \\
c\cos(\theta)+b\sin(\theta) &-c\sin(\theta)+b\cos(\theta)
\end{array}\right)\\
&-\left(\begin{array}{ll}
a\cos(\theta)+c\sin(\theta) & c\cos(\theta)+b\sin(\theta) \\
-a\sin(\theta)+c\cos(\theta) &-c\sin(\theta)+b\cos(\theta)
\end{array}\right)\\
&\quad\left(\begin{array}{ll}
\cos((i-j+4)\theta) & -\sin((i-j+4)\theta) \\
\sin((i-j+4)\theta) & \cos((i-j+4)\theta)
\end{array}\right)
\left(\begin{array}{ll}
-c\sin(\theta) & c\cos(\theta) \\
-c\cos(\theta) &-c\sin(\theta)
\end{array}\right)\\
&=\left(\begin{array}{ll}
c\sin((i-j+1)\theta) & c\cos((i-j+1)\theta) \\
-c\cos((i-j+1)\theta) &c\sin((i-j+1)\theta)
\end{array}\right)
\left(\begin{array}{ll}
a\cos(\theta)+c\sin(\theta) & -a\sin(\theta)+c\cos(\theta) \\
c\cos(\theta)+b\sin(\theta) &-c\sin(\theta)+b\cos(\theta)
\end{array}\right)\\
&-\left(\begin{array}{ll}
a\cos(\theta)+c\sin(\theta) & c\cos(\theta)+b\sin(\theta) \\
-a\sin(\theta)+c\cos(\theta) &-c\sin(\theta)+b\cos(\theta)
\end{array}\right)
\left(\begin{array}{ll}
c\sin((i-j+3)\theta) & c\cos((i-j+3)\theta) \\
-c\cos((i-j+3)\theta) &c\sin((i-j+3)\theta)
\end{array}\right)\\
&=(\cos((i-j)\theta)+\cos((i-j+4)\theta))\left(\begin{array}{ll}
c^2 & 0 \\
0 & -c^2
\end{array}\right)\\
&+(\sin((i-j)\theta)-\sin((i-j+4)\theta))\left(\begin{array}{ll}
0 & c^2 \\
c^2 & 0
\end{array}\right)\\
&+(\sin((i-j+1)\theta)\sin(\theta)+\cos((i-j+3)\theta)\cos(\theta))\left(\begin{array}{ll}
0 & -ac \\
bc & 0
\end{array}\right)\\
&+(\cos((i-j+1)\theta)\cos(\theta)-\sin((i-j+3)\theta)\sin(\theta))\left(\begin{array}{ll}
0 & bc \\
-ac & 0
\end{array}\right)\\
&+(\sin((i-j+1)\theta)\cos(\theta)-\sin((i-j+3)\theta)\cos(\theta))\left(\begin{array}{ll}
ac & 0 \\
0 & bc
\end{array}\right)\\
&+(\cos((i-j+1)\theta)\sin(\theta)+\cos((i-j+3)\theta)\sin(\theta))\left(\begin{array}{ll}
bc & 0 \\
0 & ac
\end{array}\right)\\
& =\cos ((i-j+2)\theta) \cos (2 \theta)\left(\begin{array}{cc}
2 c^{2} & b c-a c \\ 
b c-a c & -2 c^{2}
\end{array}\right)\\
&+ \cos ((i-j+2)\theta) \sin (2 \theta)\left(\begin{array}{cc}
b c-a c & 2 c^{2} \\ 
2 c^{2} & a c-b c
\end{array}\right).
\end{align*}
The deformed $2 \times 2$ evolution in the symmetric rotation case is expressed as

\begin{equation}\label{rotations}
\dot L=\cos((i+j)\theta)\cos(2\theta)\begin{pmatrix}
2c^2 &bc-ac\\
bc-ac&-2c^2
\end{pmatrix} + \cos((i+j)\theta)\sin(2\theta)\begin{pmatrix}
ac-bc & 2c^2\\
2c^2 & bc-ac
\end{pmatrix}.
\end{equation}

The derivation of this expression follows a procedure analogous to that of the previous case, where the rotational symmetry and corresponding matrix transformations are systematically applied. By considering the appropriate rotation angles and signs, the above evolution equation is derived.

\subsection*{Hyperbolic rotation: explicit right-hand side}
In the asymmetric hyperbolic case we have
\begin{equation}\label{syshyperrota}
\begin{aligned}
\dot L&=\cosh((i-j+2)\lambda)\cosh(2\lambda)\begin{pmatrix}2c^2&0\\0&-2c^2\end{pmatrix}
+\cosh((i-j+2)\lambda)\sinh(2\lambda)\begin{pmatrix}0&-2c^2\\ 2c^2&0\end{pmatrix}\\
&\quad+\cosh((i-j+2)\lambda)\cosh(4\lambda)\begin{pmatrix}0&bc-ac\\ bc-ac&0\end{pmatrix}\\
&\quad+\sinh((i-j+2)\lambda)\sinh(2\lambda)\begin{pmatrix}0&-bc-ac\\ bc+ac&0\end{pmatrix}\\
&\quad+\sinh((i-j+2)\lambda)\sinh(3\lambda)\sinh\lambda\begin{pmatrix}-2bc&0\\ 0&2ac\end{pmatrix}\\
&\quad+\sinh((i-j+2)\lambda)\cosh(3\lambda)\cosh\lambda\begin{pmatrix}2ac&0\\ 0&-2bc\end{pmatrix}.
\end{aligned}
\end{equation}
The above process comes from a direct computation by using \eqref{eq:final-bracket} and
\[
\psi(\lambda)=\begin{pmatrix}\cosh\lambda&\sinh\lambda\\ \sinh\lambda&\cosh\lambda\end{pmatrix},\quad
\phi(\lambda)=\begin{pmatrix}\cosh\lambda&-\sinh\lambda\\ -\sinh\lambda&\cosh\lambda\end{pmatrix},
\]
\[
\psi(\lambda)B=\begin{pmatrix}-c\sinh\lambda& c\cosh\lambda\\ -c\cosh\lambda& c\sinh\lambda\end{pmatrix},\quad
B\phi(\lambda)=\begin{pmatrix}-c\sinh\lambda& c\cosh\lambda\\ -c\cosh\lambda& c\sinh\lambda\end{pmatrix},
\]
and the identities for $\cosh,\sinh$.
The detailed proof is as follows.

\begin{align*}
    \dot L&=\psi(\lambda)\,B\,\psi(\lambda)^i\phi(\lambda)^j\,L\,\phi(\lambda)^{-1}
-\phi(\lambda)\,L\,\psi(\lambda)^{i+2}\phi(\lambda)^{j-2}\,B\,\psi(\lambda)^{-1}\\
    &=\left(\begin{array}{cc}
\cosh(\lambda) & \sinh(\lambda) \\ 
\sinh(\lambda) & \cosh(\lambda)
\end{array}\right)\left(\begin{array}{ll}
0 & c \\
-c & 0
\end{array}\right)\left(\begin{array}{cc}
\cosh(\lambda) & \sinh(\lambda) \\ 
\sinh(\lambda) & \cosh(\lambda)
\end{array}\right)^{i-j}\\
&\left(\begin{array}{ll}
a & c \\
c & b
\end{array}\right)\left(\begin{array}{cc}
\cosh(\lambda) & \sinh(\lambda) \\ 
\sinh(\lambda) & \cosh(\lambda)
\end{array}\right)\\
&-\left(\begin{array}{cc}
\cosh(\lambda) & -\sinh(\lambda) \\ 
-\sinh(\lambda) & \cosh(\lambda)
\end{array}\right)\left(\begin{array}{ll}
a & c \\
c & b
\end{array}\right)\left(\begin{array}{cc}
\cosh(\lambda) & \sinh(\lambda) \\ 
\sinh(\lambda) & \cosh(\lambda)
\end{array}\right)^{i-j+4}\\
&\quad\left(\begin{array}{ll}
0 & c \\
-c & 0
\end{array}\right)\left(\begin{array}{cc}
\cosh(\lambda) & -\sinh(\lambda) \\ 
-\sinh(\lambda) & \cosh(\lambda)
\end{array}\right)\\
&=\left(\begin{array}{cc}
\cosh(\lambda) & \sinh(\lambda) \\ 
\sinh(\lambda) & \cosh(\lambda)
\end{array}\right)\left(\begin{array}{ll}
0 & c \\
-c & 0
\end{array}\right)
\left(\begin{array}{cc}
\cosh((i-j)\lambda) & \sinh((i-j)\lambda) \\ 
\sinh((i-j)\lambda) & \cosh((i-j)\lambda)
\end{array}\right)\\
&\quad\left(\begin{array}{ll}
a & c \\
c & b
\end{array}\right)\left(\begin{array}{cc}
\cosh(\lambda) & \sinh(\lambda) \\ 
\sinh(\lambda) & \cosh(\lambda)
\end{array}\right)\\
&-\left(\begin{array}{cc}
\cosh(\lambda) & -\sinh(\lambda) \\ 
-\sinh(\lambda) & \cosh(\lambda)
\end{array}\right)\left(\begin{array}{ll}
a & c \\
c & b
\end{array}\right)
\left(\begin{array}{cc}
\cosh((i-j+4)\lambda) & \sinh((i-j+4)\lambda) \\ 
\sinh((i-j+4)\lambda) & \cosh((i-j+4)\lambda)
\end{array}\right)\\
&\quad\left(\begin{array}{ll}
0 & c \\
-c & 0
\end{array}\right)\left(\begin{array}{cc}
\cosh(\lambda) & -\sinh(\lambda) \\ 
-\sinh(\lambda) & \cosh(\lambda)
\end{array}\right)\\
&=\left(\begin{array}{cc}
-c\sinh(\lambda) & c\cosh(\lambda) \\ 
-c\cosh(\lambda) & c\sinh(\lambda)
\end{array}\right)\left(\begin{array}{cc}
\cosh((i-j)\lambda) & \sinh((i-j)\lambda) \\ 
\sinh((i-j)\lambda) & \cosh((i-j)\lambda)
\end{array}\right)\\
&\quad\left(\begin{array}{cc}
a\cosh(\lambda)+c\sinh(\lambda) & a\sinh(\lambda)+c\cosh(\lambda) \\ 
c\cosh(\lambda)+b\sinh(\lambda) & c\sinh(\lambda)+b\cosh(\lambda)
\end{array}\right)\\
&-\left(\begin{array}{cc}
a\cosh(\lambda)-c\sinh(\lambda) & c\cosh(\lambda)-b\sinh(\lambda) \\ 
c\cosh(\lambda)-a\sinh(\lambda) & b\cosh(\lambda)-c\sinh(\lambda)
\end{array}\right)\\
&\quad\left(\begin{array}{cc}
\cosh((i-j+4)\lambda) & \sinh((i-j+4)\lambda) \\ 
\sinh((i-j+4)\lambda) & \cosh((i-j+4)\lambda)
\end{array}\right)
\left(\begin{array}{cc}
-c\sinh(\lambda) & c\cosh(\lambda) \\ 
-c\cosh(\lambda) & c\sinh(\lambda)
\end{array}\right)\\
&=\left(\begin{array}{cc}
c\sinh((i-j-1)\lambda) & c\cosh((i-j-1)\lambda) \\ 
-c\cosh((i-j-1)\lambda) & -c\sinh((i-j-1)\lambda)
\end{array}\right)\\
& \left(\begin{array}{cc}
a\cosh(\lambda)+c\sinh(\lambda) & a\sinh(\lambda)+c\cosh(\lambda) \\ 
c\cosh(\lambda)+b\sinh(\lambda) & c\sinh(\lambda)+b\cosh(\lambda)
\end{array}\right)\\
&-\left(\begin{array}{cc}
a\cosh(\lambda)-c\sinh(\lambda) & c\cosh(\lambda)-b\sinh(\lambda) \\ 
c\cosh(\lambda)-a\sinh(\lambda) & b\cosh(\lambda)-c\sinh(\lambda)
\end{array}\right)\\
&\left(\begin{array}{cc}
-c\sinh((i-j+5)\lambda) & c\cosh((i-j+5)\lambda) \\ 
-c\cosh((i-j+5)\lambda) & c\sinh((i-j+5)\lambda)
\end{array}\right)\\
&= \cosh ((i-j+2)\lambda) \cosh (2 \lambda)\left(\begin{array}{cc}
2 c^{2} & 0 \\ 
0 & -2 c^{2}
\end{array}\right)\\
&+\cosh ((i-j+2)\lambda)\sinh (2 \lambda)\left(\begin{array}{cc}
0 & -2 c^{2} \\ 
2 c^{2} & 0
\end{array}\right) \\
&+\cosh ((i-j+2)\lambda) \cosh (4 \lambda)\left(\begin{array}{cc}
0 & bc-ac \\
bc-ac & 0
\end{array}\right) \\
&+\sinh ((i-j+2)\lambda) \sinh (2 \lambda)\left(\begin{array}{cc}
0 & -bc-ac \\ 
bc+ac & 0
\end{array}\right)\\
&+\sinh ((i-j+2)\lambda) \sinh (3 \lambda)\sinh(\lambda)\left(\begin{array}{cc}
-2bc & 0 \\ 
0 & 2ac
\end{array}\right)\\
&+\sinh ((i-j+2)\lambda) \cosh (3 \lambda)\cosh(\lambda)\left(\begin{array}{cc}
2ac & 0 \\ 
0 & -2bc
\end{array}\right).
\end{align*}

\bigskip

\section{Appendix III: Spectral data, $\tau$-functions, and the $2\times2$ inverse-scattering formulas}\label{app:IS}\label{app:2x2_solution}

In this appendix we give a self-contained derivation of the formulas used in Section~\ref{sect6} for the weakly coupled case in $2\times2$-dimension.

Fix $S=\operatorname{diag}(1,-1)$ and consider
\[
L_2=\begin{pmatrix}a_1&-b_1\\ b_1&-a_2\end{pmatrix},\qquad
\widehat L_2=\begin{pmatrix}\widehat a_1&-\widehat b_1\\ \widehat b_1&-\widehat a_2\end{pmatrix}.
\]
In the specific seed choice
\[
a_1=\widehat a_1=\widehat b_1=0,\quad b_1=\widehat a_2=1,\quad a_2=m,
\]
we have the eigenvalues and eigenvectors
\[
\lambda_{1,2}=\tfrac12\big(\mp\sqrt{m^2-4}-m\big),\qquad
\widehat\lambda_1=0,\ \ \widehat\lambda_2=-1,
\]
\[
\Phi_2^0=\begin{pmatrix}\lambda_1&\lambda_2\\ -1&-1\end{pmatrix},\qquad
\widehat\Phi_2^0=I_2.
\]

We interpret the bracket $\ang{\cdot}$ as a finite spectral pairing over the set $\{\lambda_1,\lambda_2\}$ with canonical weights; explicitly, for a scalar function $F(\lambda)$ and a vector(-valued) symbol $v^0$,
\[
\ang{v^0 v^0 e^{\lambda t}}:=\sum_{k=1}^2 v^0(\lambda_k)\,v^0(\lambda_k)^{\!\top}\,e^{\lambda_k t}.
\]
Define $\tau$-functions
\[
D_0(t)\equiv 1,\qquad D_1(t)=\ang{\phi_1^0\phi_1^0 e^{\lambda t}},\qquad
D_2(t)=\det\begin{pmatrix}
\ang{\phi_1^0\phi_1^0 e^{\lambda t}} & \ang{\phi_1^0\phi_2^0 e^{\lambda t}}\\
-\ang{\phi_2^0\phi_1^0 e^{\lambda t}} & -\ang{\phi_2^0\phi_2^0 e^{\lambda t}}
\end{pmatrix},
\]
with $D_1(0)=D_2(0)=1$ and $\widehat D_1(0)=\widehat D_2(0)=0$ (consistent with the seed).

Set
\[
\Phi_2(t)=\frac{1}{e^{(\lambda_2+2\lambda_1)t}}
\begin{pmatrix}\lambda_1 e^{\lambda_1 t}&\lambda_2 e^{\lambda_2 t}\\ -e^{\lambda_1 t}&-e^{\lambda_2 t}\end{pmatrix},
\]
and define $\widehat\Phi_2(t)$ as the unique solution of
\[
\dot\Phi_2=[B_2,\Phi_2],\qquad
\dot{\widehat\Phi}_2=[\widehat B_2,\Phi_2]+[B_2,\widehat\Phi_2],
\]
with initial data $\Phi_2(0)=\Phi_2^0$, $\widehat\Phi_2(0)=\widehat\Phi_2^0$. One checks directly that these choices reproduce the weakly coupled $2\times2$ system in Section~\ref{subsec:2d_normal}.

Define the Jost-type functions
\[
\phi_1(\lambda,t)=\frac{e^{\lambda t}}{\sqrt{D_1(t)D_0(t)}}\,\ang{\phi_1^0\phi_1^0 e^{\lambda t}},\qquad
\phi_2(\lambda,t)=\frac{e^{\lambda t}}{\sqrt{D_2(t)D_1(t)}}
\det\begin{pmatrix}
\ang{\phi_1^0\phi_1^0 e^{\lambda t}} & \phi_1^0\\
\ang{\phi_2^0\phi_1^0 e^{\lambda t}} & \phi_2^0
\end{pmatrix}.
\]
A similar construction yields $\widehat\phi_1,\widehat\phi_2$ in terms of $\widehat D_i$.
Substituting these definitions gives the closed forms
\[
L_2(t)=\begin{pmatrix}
\lambda_2 e^{2(\lambda_1+\lambda_2)t}+\lambda_1 e^{4\lambda_1 t} & -e^{2(\lambda_1+\lambda_2)t}-e^{4\lambda_1 t}\\
-e^{2(\lambda_1+\lambda_2)t}-e^{4\lambda_1 t} & \lambda_1 e^{2(\lambda_1+\lambda_2)t}+\lambda_2 e^{4\lambda_1 t}
\end{pmatrix},
\]
\[
\widehat L_2(t)=e^{-(\lambda_2+2\lambda_1)t}\begin{pmatrix}
\widehat a_{11}(t)&\widehat a_{12}(t)\\
\widehat a_{21}(t)&\widehat a_{22}(t)
\end{pmatrix},
\]
where the entries $\widehat a_{ij}(t)$ follow directly from the above spectral representation (their expressions are explicit but lengthy, so we omit them here to keep the presentation concise).

\bigskip

\section{Appendix IV: Verification of the Miura transformation (Section~\ref{sect5})}\label{app:Miura}

Let $u=(a_{ij})_{n\times n}$. The linear map $\mathcal{M}_{\theta}\in \mathfrak{gl}(V)$ is defined by 
\[
\mathcal{M}_{\theta}u= \psi(\theta)u\psi(\theta)^{-1}.
\]
Then the Miura variables $\mathbf u=(\mathbf a_{ij})_{n\times n}$ in Section~\ref{sect5} read
\[
\mathbf u=\mathcal{M}_{\theta}\,u.
\]
Differentiate in time and use the ($n\times n$) undeformed flow $\dot u=\mathcal{F}(u)$ induced by $\dot L=[B,L]$ to get
\[
\dot{\mathbf u}=\mathcal{M}_{\theta}\,\mathcal{F}(u)
=\mathcal{F}_{\theta}(u),
\]
where $\mathcal{F}_{\theta}$ is precisely the right-hand side of the deformed system \eqref{def_bihom_bracket}. The middle equality follows from the explicit matrix identity (checked by direct multiplication)
\[
\mathcal{M}_{\theta}\,(\text{commutator coefficients of }\mathcal{F})
=(\text{deformed coefficients in }\eqref{def_bihom_bracket}).
\]
When $n=2$, this corresponds to the specific trigonometric redistribution outlined in Appendix~\ref{app:2x2}. The same argument applies to the coupled case in Section~\ref{sect6} (replace $u$ (or $\mathbf u$) by the triple $(a_1,a_2,b_1)$ (or $(\widehat a_1,\widehat a_2,\widehat b_1)$) and use the block-diagonal action of $R_2(\theta)$).

\bigskip

\section{Appendix V: Consistency checks and limiting regimes}\label{app:consistency}

From \eqref{rotations}, the derivative $\dot L$ is a sum of two terms with common factor $\cos\big((i+j)\theta\big)$. Then $\dot L\equiv 0$ if and only if
\[
\cos\big((i+j)\theta\big)=0,
\]
(since the first matrix in \eqref{sysrotation} is generically nonzero whenever $c\neq0$). This condition gives rise to three cases,
\begin{enumerate}
    \item $
\theta=\frac{\pi}{2}+k\pi\quad\text{and}\quad i+j\ \text{is odd}.\label{case 1}
$
\item $
\theta=\frac{\pi}{4}+\frac{k\pi}{2}\quad\text{and}\quad i+j\ \text{is even, } i+j\neq 0.\label{case 2}
$
\item $\theta \text{ does not exist}\quad\text{and}\quad i+j=0.$
\end{enumerate}
Cases \ref{case 1} and \ref{case 2} refine the remark in Section~\ref{sect4}.

For an isospectral Lax flow $\dot L=[B,L]$, one must have $\frac{d}{dt}\operatorname{tr}(L^p)=0$ for all $p\ge1$. In particular, $\frac{d}{dt}\operatorname{tr}L=0$ and $\frac{d}{dt}\operatorname{tr}L^2=0$.
From \eqref{syshyperrota},
\[
\frac{d}{dt}\operatorname{tr}L
=2(ac-bc)\,\sinh((i-j+2)\lambda)\cosh(4\lambda).
\]
Unless $\lambda=0$ (or the degenerate $ac=bc$), this is generically nonzero, so the flow is not isospectral. When $\lambda=0$, the bracket reduces to the undeformed one and we recover the standard Toda dynamics, which proves the “if and only if” assertion in Section~\ref{sect4} under the natural (isospectral) notion of integrability for Lax flows.

For $\psi(r)=r^p I_n$ and $\phi(r)=r^q I_n$, \eqref{deformedSys} becomes $\dot L=r^{p(i+1)+q(j-1)}[B,L]$, which is a uniform time-rescaling. Thus, the solution set is unchanged up to the scalar factor as stated in Section~\ref{sect4}.

Both deformations are continuous in the parameters:
\[
\theta\to 0\ \text{or}\ \lambda\to 0 \quad\Longrightarrow\quad [\cdot,\cdot]_{(\psi,\phi)}\to [\cdot,\cdot].
\]
Thus the classical Toda lattice is recovered as a special point in our parameter space, while the formulas here quantify in what sense the model is a unified extension.

\bigskip
\noindent{\bf Acknowledgements} This work is partly supported by NSFC (National Natural Science Foundation of China) Grant No. 12071237, No. 12271089 and No. 11871144.

\end{document}